\documentclass[journal]{IEEEtran}
\usepackage{amsfonts}
\usepackage{amssymb}
\usepackage[cmex10]{amsmath}
\usepackage{graphicx}
\usepackage{subfigure}
\usepackage{verbatim}
\usepackage{cite}
\usepackage{diagbox}
\usepackage{tabularx,booktabs}
\usepackage{amsthm}
\usepackage{booktabs}
\usepackage[ruled,lined]{algorithm2e}
\usepackage{multirow}
\usepackage{setspace}
\usepackage{stfloats}
\usepackage{fancyhdr}
\usepackage{float}
\usepackage{algorithmic}
 \usepackage{enumerate}

\IEEEoverridecommandlockouts
\begin{document}
%\begin{CJK}{GBK}{kai}

\newcommand{\be}{\begin{equation}}
\newcommand{\ee}{\end{equation}}
\newcommand{\br}{{\mbox{\boldmath{$r$}}}}
\newcommand{\bp}{{\mbox{\boldmath{$p$}}}}
\newcommand{\bpi}{\mbox{\boldmath{ $\pi $}}}
\newcommand{\bn}{{\mbox{\boldmath{$n$}}}}
\newcommand{\balfa}{{\mbox{\boldmath{$\alpha$}}}}
\newcommand{\ba}{\mbox{\boldmath{$a $}}}
\newcommand{\bta}{\mbox{\boldmath{$\beta $}}}
\newcommand{\bg}{\mbox{\boldmath{$g $}}}
\newcommand{\bPsi}{\mbox{\boldmath{$\Psi $}}}
\newcommand{\bsigma}{\mbox{\boldmath{ $\Sigma $}}}
\newcommand{\bGamma}{{\bf \Gamma }}
\newcommand{\bA}{{\bf A }}
\newcommand{\bP}{{\bf P }}
\newcommand{\bX}{{\bf X }}
\newcommand{\bI}{{\bf I }}
\newcommand{\bR}{{\bf R }}
\newcommand{\bZ}{{\bf Z }}
\newcommand{\bz}{{\bf z }}
\newcommand{\bx}{{\mathbf{x}}}
\newcommand{\bM}{{\bf M}}
\newcommand{\bU}{{\bf U}}
\newcommand{\bD}{{\bf D}}
\newcommand{\bJ}{{\bf J}}
\newcommand{\bH}{{\bf H}}
\newcommand{\bK}{{\bf K}}
\newcommand{\bm}{{\bf m}}
\newcommand{\bN}{{\bf N}}
\newcommand{\bC}{{\bf C}}
\newcommand{\bL}{{\bf L}}
\newcommand{\bF}{{\bf F}}
\newcommand{\bv}{{\bf v}}
\newcommand{\bSigma}{{\bf \Sigma}}
\newcommand{\bS}{{\bf S}}
\newcommand{\bs}{{\bf s}}
\newcommand{\bO}{{\bf O}}
\newcommand{\bQ}{{\bf Q}}
\newcommand{\btr}{{\mbox{\boldmath{$tr$}}}}
\newcommand{\bNSCM}{{\bf NSCM}}
\newcommand{\barg}{{\bf arg}}
\newcommand{\bmax}{{\bf max}}
\newcommand{\test}{\mbox{$
\begin{array}{c}
\stackrel{ \stackrel{\textstyle H_1}{\textstyle >} } { \stackrel{\textstyle <}{\textstyle H_0} }
\end{array}
$}}
\newcommand{\tabincell}[2]{\begin{tabular}{@{}#1@{}}#2\end{tabular}}
\newtheorem{Def}{Definition}
\newtheorem{Pro}{Proposition}
\newtheorem{Exa}{Example}
\newtheorem{Rem}{Remark}

\title{Distributed Multi-Sensor Fusion Using Generalized Multi-Bernoulli Densities}
\author{\IEEEauthorblockN{Meng Jiang*, Wei Yi*, Reza Hoseinnezhad$^\dag$ and Lingjiang~Kong* }\\
\IEEEauthorblockA{
*University of Electronic Science and Technology of China\\
Email:\{kussoyi, lingjiang.kong\}@gmail.com \\
$^\dag$RMIT University, Australia, Email: rezah@rmit.edu.au
}}
%School of Electronic Engineering\\

\maketitle
 \thispagestyle{empty}
\begin{abstract}
The paper addresses distributed multi-target tracking in the framework of generalized Covariance Intersection (GCI) over multistatic radar system. The proposed method is based on the unlabeled version of generalized labeled multi-Bernoulli (GLMB) family by discarding the labels, referred as generalized multi-Bernoulli (GMB) family. However, it doesn't permit closed form solution for GCI fusion with GMB family.
To solve this challenging problem, firstly, we propose an efficient approximation to the GMB family which preserves both the probability hypothesis density (PHD) and cardinality distribution, named as second-order approximation of GMB (SO-GMB) density. Then, we derive explicit expression for the GCI fusion with SO-GMB density.
Finally, we compare the first-order approximation of GMB (FO-GMB) density with SO-GMB density in two scenarios and make a concrete analysis of the advantages of the second-order approximation.
Simulation results are presented to verify the proposed approach.
\end{abstract}

\section{Introduction}
Due to their scalability, flexibility, robustness and fault-tolerance, distributed signal processing methods provide unique advantages over similar but centralized techniques. A particular family of such methods are the \textit{distributed sensor fusion} techniques used for multi-object estimation, or more specifically, multi-target tracking (MTT). These techniques are at the core of various distributed multi-sensor MTT systems such as multistatic radar systems. Those systems have received significant attention in the last decade, and have been used in a wide range of applications from traffic monitoring to battlefield surveillance.

A distributed sensor fusion solution for multi-target tracking usually includes two major components: (i) an efficient and robust multi-target  filter to run locally in each node of the sensor network, independent of the network structure, and (ii) an algorithm for distributed fusion of the information received by each node from multiple other nodes, that includes unknown level of correlation.

In terms of choosing the multi-target filter component of an optimal distributed multi-sensor multi-target tracking system, we note the recent development of the notion of labeled random finite sets and their associated filters~\cite{vo2013labeled}. Filters such as Labeled Multi-Bernoulli (LMB) filter~\cite{reuter2014labeled} and Vo-Vo filter (also called Generalized Labeled Multi-Bernoulli (GLMB) filter)\footnote{We follow Malher who named Vo-Vo filter in his book~\cite{mahler2014advances} for the first time.} have been of significant interest due to their superior performance in terms of accuracy of cardinality and state estimation as well as tracking multiple trajectories. Vo-Vo filter~\cite{vo2013labeled,vo2014labeled} provides a closed-form solution to the optimal Bayesian filter, and has shown to outperform the well-known Probability Hypothesis Density (PHD) filter and its cardinalized version, CPHD filter, and the multi-Berboulli (MB) filter in challenging multi-target tracking scenarios.

This paper focuses on the sensor fusion component of the distributed MTT solution, with particular interest in Vo-Vo filter to be chosen as the local multi-target filter running in each node of the sensor network. An effective information fusion algorithm is expected to combine the information generated by a number of sensor (e.g. radar) nodes and to achieve state estimates and target tracks that are in maximum consistence with \textit{all} the information obtained from multi-sensor measurements. Because of the unacceptable cost of computing the common information between nodes, optimal fusion~\cite{chong1990distributed} is ruled out and one needs to resort to robust suboptimal fusion rules. Mahler \cite{mahler2000optimal} proposed the Generalized Covariance Intersection (GCI) fusion rule based on Exponential Mixture Densities (EMDs). Using this rule, both Gaussian and non-Gaussian formed multi-target distributions from different radars with completely unknown correlation, can be fused sub-optimally.

Following the introduction of GCI fusion rule by Mahler, Clark \textit{et al}.~\cite{clark2010robust} developed tractable formulations for GCI-based fusion of  multi-target posteriors. This work was followed by particle implementation~\cite{uney2013distributed} and Gaussian mixture implementation~\cite{battistelli2013consensus} of distributed fusion of Poisson posteriors (suitable for solution designs involving PHD filters working in each node), then a distributed track-before-detect (TBD) solution based on using local Bernoulli filtering of measurements provided by a Doppler-shift sensor network~\cite{guldogan2014consensus}. Wang \textit{et al}.~\cite{wang2015distributedf} recently presented a distributed fusion method with multi-Bernoulli (MB) filters~\cite{vo2009cardinality, dunne2012member, vo2010joint, lian2012convergence, wei2009mobile, wei2010sensor,hoseinnezhad2012visual} working in each node, based on GCI rule.

With the recent development of labeled set filters and their advantageous performance compared to previous (unlabeled) random set filters, the design of new distributed sensor fusion systems (with labeled set filters working in each node) is of both fundamental and practical interest. The major task here is to develop tractable algorithms for sufficiently accurate GCI-based fusion of labeled random set posteriors. This is a challenging task because of the \textit{label space mismatching phenomenon}; the same realization can be drawn from label spaces of different sensor nodes, which do not have the same implication. To tackle this problem,
%Fantacci~\textit{et al}.~\cite{fantacci2015consensus} propose a label space alignment technique to find consistent labels before GCI-based fusion of GLMB posteriors.
based on the assumption that all the sensor nodes share the same label space for the birth process, Fantacci~\textit{et al}.~\cite{fantacci2015consensus} implementes the GCI fusion with labeled set filters by using the consistent label directly.
In \cite{fantacci2015consensus}, analytic formulas for distributed fusion using labeled RFSs are presented and the approach performers well in the situation where the label spaces of each radar node are matching.

An alternative approach is to perform GCI-based fusion with the \textit{unlabeled} versions of the distributions in the GLMB family, named as generalized multi-Bernoulli (GMB) family. Wang~\textit{et al}.~\cite{wang2015distributed,GCI-LM} proposed a tractable solution to the GCI fusion of GMB posteriors via approximating each with a multi-Bernoulli (MB) distribution that matches its first-order moment (its PHD). The approximate MB distribution was also referred to as the first-order approximation of GMB (FO-GMB) distribution.

In this paper, we focus on the second fusion approach and address the problem of the distributed GCI-based fusion with labeled set filters.  Inspired by the approach through which PHD filter was extended to CPHD, and LMB to M$\delta$-GLMB, we present a second order approximation to a GMB density (SO-GMB) that matches not only its PHD but also its cardinality distribution. Just as the CPHD and M$\delta$-GLMB filters perform better than PHD and LMB filters (because both preserve the second-order characteristics), we expect that distributed fusion of SO-GMBs performs better than FO-GMBs. We formulate a tractable GCI-based fusion rule for SO-GMB densities. The fused posterior turns out to be another GMB distribution, and the formula enables sequential fusion within a network of more than two radar nodes. We analyze the performance of GCI-based fusion with SO-GMB densities in different application scenarios. The simulation results indicate that while the performance of SO-GMB fusion is slightly better than FO-GMB fusion in simple tracking scenarios, in challenging situations where the targets are crossing and move in close proximity, SO-GMB fusion significantly outperforms FO-GMB fusion.

\section{Background}
This section provides a brief review of GLMB RFS, GMB RFS and GCI fusion necessary for the results of this paper. For further details, we refer the reader to \cite{wang2015distributed}.
\subsection{Notation }
In this paper, we adhere to the convention that single-target states are denoted by the small letters, e.g.,$x, \mathbf{x}$ while multi-target states are denoted by capital letters, e.g.,$X, \mathbf{X}$. Symbols for labeled states and their distributions/statistics (single-target or multi-target) are bolded to distinguish them from unlabeled ones, e.g., $\mathbf{x}, \mathbf{X}, \bpi$, etc. To be more specific, the labeled single target state $\bx$ is constructed by augmenting a state $x\in\mathbb{X}$ with a label $\ell\in\mathbb{L}$. Observations generated by single-target states are denoted by the small letter, e.g., $z$, and the multi-target observations are denoted by the capital letter, e.g., $Z$.
%The labels are usually drawn from a discrete label space, $\mathbb{L}=\{\alpha_i:i\in\mathbb{N}\}$, where all $\alpha_i$ are distinct and the index space $\mathbb{N}$ is the set of positive integers.
Additionally, blackboard bold letters represent spaces, e.g., the state space is represented by $\mathbb{X}$, the label space by $\mathbb{L}$, and the observation space by $\mathbb{Z}$. The collection of all finite sets of $\mathbb{X}$ is denoted by $\mathcal{F}(\mathbb{X})$ and $\mathcal{F}_n(\mathbb{X})$ denotes all finite subsets with $n$ elements.

Moreover, in order to support arbitrary arguments like sets, vectors and integers, the generalized Kronecker delta function is given by
\begin{equation}\label{delta}
  \delta_Y(X)\triangleq\left\{\begin{array}{l}
\!\!1, \,\,\,\,\mbox{if $X = Y$} \\
\!\!0, \,\,\,\,\mbox{otherwise}
\end{array}\right.
\end{equation}
and the inclusion function is given by
\begin{equation}\label{inclusion function}
  1_Y(X)\triangleq\left\{\begin{array}{l}
\!\!1, \,\,\,\,\mbox{if $X \subseteq Y$}\\
\!\!0. \,\,\,\,\mbox{otherwise}
\end{array}\right.
\end{equation}

\subsection{GLMB RFS }
The GLMB multi-object distribution was recently formulated and introduced by  Vo and Vo~\cite{vo2013labeled}. Approximating the prior by this general type of distribution for labeled multiple objects forms the basis of an analytic solution to the Bayes multi-object filter. Under the standard multi-object model, the GLMB is closed under the Chapman-Kolmogorov equation and is also \textit{a conjugate prior} with the well-known point measurement likelihood. Thus, with a GLMB prior, the predicted and posterior densities are guaranteed to be GLMB as well.

Let $\mathcal{L}:\mathbb{X}\times\mathbb{L}\rightarrow\mathbb{L}$ be the projection $\mathcal{L}((x,\ell))=\ell$, and $\Delta(\bX)=\delta_{|\bX|}(|\mathcal{L}(\bX)|)$ denote the distinct label indicator.
A GLMB is a RFS on $\mathbb{X}\times\mathbb{L}$ distributed according to
\begin{equation}\label{GLMB}
\bpi(\bX)=\Delta(\bX)\sum_{c\in\mathbb{C}}w^{(c)}(\mathcal{L}(\bX))[p^{(c)}]^{\bX}
\end{equation}
where $\mathbb{C}$ is a discrete index set. The weights $w^{(c)}(L)$ and the spatial distributions $p^{(c)}$ satisfy the normalization conditions
\begin{equation}\label{GLMB_1}
\sum_{L\subseteq\mathbb{L}}\sum_{c\in\mathbb{C}}w^{(c)}(L)=1
\end{equation}
\begin{equation}\label{GLMB_2}
\,\,\,\,\,\int p^{(c)}(x,\ell)dx=1.
\end{equation}
The GLMB density (\ref{GLMB}) can be interpreted as a mixture of multi-object exponentials.
\subsection{GMB RFS }
The unlabeled version of GLMB, called Generalized Multi-Bernoulli (GMB) is a multi-object distribution defined in the state space $\mathbb{X}$ and given by~\cite{wang2015distributed}:
\begin{equation}\label{GMB}
\begin{split}
\pi&(\left\{x_1,\ldots,x_n\right\})= \\
&\sum_\sigma\sum_{(\mathcal{I},\phi)\in\mathcal{F}_n(\mathbb{I})\times\Phi}w^{\left(\mathcal{I},\phi\right)}
\prod_{i=1}^{n}p^{\left(\phi\right),\mathcal{I}^v\left(i\right)}\left(x_{\sigma(i)}\right)
\end{split}
\end{equation}
where the summation $\sum_\sigma$ is taken over all permutations on the number $1,\cdots,n$, $\Phi$ is a discrete space, $\mathbb{I}$ is the index set of densities, $\mathcal{I}^v\in\mathbb{I}^{\left|\mathcal{I}\right|}$ is a vector constructed by sorting the elements of set $\mathcal{I}$, $w^{\left(\mathcal{I},\phi\right)}$ and $p^{\left(\phi\right),\imath}(x)$ satisfy
\begin{equation}\label{GMB_1}
\sum_{(\mathcal{I},\phi)\in\mathcal{F}(\mathbb{I})\times\Phi}w^{\left(\mathcal{I},\phi\right)}=1
\end{equation}
\begin{equation}\label{GMB_2}
\,\,\,\,\,\,\,\,\,\,\,\,\,\,\,\,\,\,\,\,\,\,\,\,\,\,\,\,\int{p^{\left(\phi\right),\imath}(x)}dx=1,\imath\in\mathbb{I}.
\end{equation}
The cardinality distribution of GMB RFS is given by
\begin{equation}\label{GMB_pn}
\rho(n)=\sum_{(\mathcal{I},\phi)\in\mathcal{F}_n(\mathbb{I})\times\Phi}w^{\left(\mathcal{I},\phi\right)}.
\end{equation}
Accordingly, the PHD becomes
\begin{equation}\label{GMB_phd}
\begin{split}
v(x)&=\sum_{(\mathcal{I},\phi)\in\mathcal{F}(\mathbb{I})\times\Phi}w^{\left(\mathcal{I},\phi\right)}
\sum_{\imath\in\mathcal{I}}p^{\left(\phi\right),\imath}(x)  \\
&=\sum_{\imath\in\mathbb{I}}\sum_{(\mathcal{I},\phi)\in\mathcal{F}(\mathbb{I})\times\Phi}1_\mathcal{I}(\imath)
w^{\left(\mathcal{I},\phi\right)}p^{\left(\phi\right),\imath}(x).
\end{split}
\end{equation}
\subsection{GCI Fusion Rule }
The GCI fusion rule was proposed by Mahler \cite{mahler2000optimal} specifically to enable the fusion of FISST densities in distributed multi-sensor fusion applications. In a network of sensors (e.g. a radar network), consider the point measurement sets returned by the sensors at two nodes, and denote them by $Z_1^k$ and $Z_2^k$ where $k$ is the current time. Let us also denote the measurement history at node $i$ by $Z_i^{1:k}=(Z_i^{1},\ldots,Z_i^{k})$, and the current multi-object state by $X^k=\{x^k_1,\ldots,x^k_{n}\}$.

At time $k$, node 1 maintains computes its own local multi-object posterior $\pi_{1}(X^k|Z_{1}^{1:k})$ through an update step, and receives the locally updated posteriors from a number of other (probably neighboring) nodes in the network. The fusion problem is how to optimally combine the multiple posteriors so that maximum information is preserved in the fused posterior. We are specially interested in a fusion rule that can be implemented in a sequential way. Thus, if one of the nodes communicating its local posterior with node 1 is node 2, the fusion problem is reduced to the problem of combining $\pi_{1}(X^k|Z_{1}^{1:k})$ and $\pi_{2}(X^k|Z_{2}^{1:k})$.

According to the GCI fusion rule, a sub-optimal fused distribution is given by the following geometric mean (or exponential mixture) of the local posteriors,
\begin{equation} \label{GCI}
\begin{array}{c}
\pi_{\omega}(X^k|Z_1^{1:k},Z_{2}^{1:k}) = \frac{\pi_{1}(X^k|Z_{1}^{1:k})^{\omega_1}\pi_{2}(X^k|Z_{2}^{1:k})^{\omega_2}}
                                              {\int \pi_{1}(X^k|Z_{1}^{1:k})^{\omega_1}\pi_{2}(X^k|Z_{2}^{1})^{\omega_2}\delta X}
\end{array}
\end{equation}
where $\omega_1$, $\omega_2$ ($\omega_1+\omega_2=1$) are the parameters determining the relative fusion weight of each nodes, and the set integral is computed according to~\cite{mahler2014advances}
\begin{equation*}\label{set integral}
  \int\! f(X)\delta X\!=\sum_{n=0}^\infty \frac{1}{n!}\int\! f(\{x_1,\cdots,x_n\})dx_1\cdots dx_n.
\end{equation*}
It has been shown that among all exponential mixture densities (EMDs), the density given by equation~(\ref{GCI}) minimizes the following weighted sum of distances,
\begin{equation}\label{EMD}
  \pi_\omega=\arg \min_\pi(\omega_1D(\pi\parallel \pi_1)+\omega_2 D(\pi\parallel \pi_2))
\end{equation}
where $D$ denotes Kullback-Leibler divergence (KLD)~\cite{battistelli2013consensus}.
For convenience, in what follows we omit the conditioning on the observations and time index $k$.

\section{GCI-Based Distributed Fusion}
Consider a sensor network used in a multi-target tracking application, and a distributed sensor fusion and tracking algorithm in which a Vo-Vo filter (with GLMB assumption for underlying multi-target distributions) is run locally in each node of the network. To fuse the posteriors coming from neighboring sensor nodes, we assume that local GLMB posteriors are turned into unlabeled GMB densities before fusion. As it was mentioned earlier, Wang \textit{et al}.~\cite{wang2015distributed} proposed an approximate solution in which each GMB distribution was replaced with an MB density with the same first-order statistical moment. The approximate MB density was called a first order approximation to the GMB density (FO-GMB density) and a closed-form GCI rule-based fusion formula was derived to combine two FO-GMB densities. The fused density was shown to turn into a GMB itself.

In this paper, we follow the same approach explained above, however, we introduce the novel idea that instead of approximating a GMB with an MB density with matching PHD, a better approximation can achieved when both the first moment and the entire cardinality distribution are matched. We call such an approximate density as the second-order GMB (SO-GMB) approximate.

\begin{Def}
Consider the GMB density $\pi$ given in (\ref{GMB}). A SO-GMB density $\hat{\pi}$ corresponding to $\pi$ is given by:
\begin{equation}\label{SO_GMB}
\hat{\pi}(\left\{x_1,\ldots,x_n\right\})=\sum_\sigma\sum_{\mathcal{I}\in\mathcal{F}_n(\mathbb{I})}\hat{w}^{\left(\mathcal{I}\right)}
\prod_{i=1}^{n}\hat{p}^{\mathcal{I}^v\left(i\right)}\left(x_{\sigma(i)}\right)
\end{equation}
where
\begin{eqnarray}
\label{SO_GMB_1}
\hat{w}^{\left(\mathcal{I}\right)}&=&\sum_{\phi\in\Phi}w^{\left(\mathcal{I},\phi\right)} \\
\label{SO_GMB_2}
\hat{p}^{\imath}\left(x\right)&=&\frac{1}{\hat{w}^{\left(\mathcal{I}\right)}}
\sum_{\phi\in\Phi}w^{\left(\mathcal{I},\phi\right)}p^{\left(\phi\right),\imath}(x)
\end{eqnarray}
\end{Def}

\begin{Pro}
The SO-GMB density in (\ref{SO_GMB})-(\ref{SO_GMB_2}) preserves both PHD and cardinality distribution of the original GMB density in (\ref{GMB}).
\end{Pro}

\begin{proof}
The cardinality distribution of the Marginalized GMB density becomes
\begin{equation}\label{SO_GMB_pn}
\begin{split}
\hat{\rho}(n)&=\frac{1}{n!}\int\hat{\pi}(\left\{x_1,\ldots,x_n\right\})d(\left\{x_1,\ldots,x_n\right\}) \\
&=\frac{1}{n!}\sum_\sigma\sum_{\mathcal{I}\in\mathcal{F}_n(\mathbb{I})}\hat{w}^{\left(\mathcal{I}\right)} \\
&=\sum_{\mathcal{I}\in\mathcal{F}_n(\mathbb{I})}\hat{w}^{\left(\mathcal{I}\right)} \\
&=\sum_{(\mathcal{I},\phi)\in\mathcal{F}_n(\mathbb{I})\times\Phi}w^{\left(\mathcal{I},\phi\right)}
\end{split}
\end{equation}
which matches the cardinality distribution $\rho(n)$ given in equation~(\ref{GMB_pn}).

\noindent The PHD is given by
\begin{equation}
\nonumber
\begin{array}{l}
v(x_1)=  \sum_{n=0}^{\infty}\frac{1}{n!}\int\hat{\pi}(\{x_1,x_2,\ldots,x_{n+1}\})\ dx_2\ \ldots\ dx_{n+1}\\  \\
\ \ \ \ =  \sum_{n=0}^{\infty}\frac{1}{n!}\int
\sum_\sigma\sum_{\mathcal{I}\in\mathcal{F}_{n+1}(\mathbb{I})}\hat{w}^{\left(\mathcal{I}\right)}
\prod_{i=1}^{n+1}\hat{p}^{\mathcal{I}^v\left(i\right)}\left(x_{\sigma(i)}\right)\
\\
\hspace{6.2cm} dx_2\ \cdots\ dx_{n+1}.
\end{array}
\end{equation}
After substituting the weight terms $\hat{w}^{\left(\mathcal{I}\right)}$ with their equivalent from equation~(\ref{SO_GMB_1}), and the density terms $\hat{p}^{\mathcal{I}^v\left(i\right)}\left(x_{\sigma(i)}\right)$ with their equivalents from equation~(\ref{SO_GMB_2}), i.e.
\begin{equation}
\nonumber
\begin{array}{rcl}
\hat{w}^{\left(\mathcal{I}\right)}&=&\sum_{\phi'\in\Phi}w^{\left(\mathcal{I},\phi'\right)} \\
&& \\
\hat{p}^{\mathcal{I}^v\left(i\right)}\left(x_{\sigma(i)}\right))&=&\dfrac{
\sum_{\phi''\in\Phi}w^{\left(\mathcal{I},\phi''\right)}p^{\left(\phi''\right),\mathcal{I}^v\left(i\right)}(x_{\sigma(i)})
}
{\sum_{\phi'\in\Phi}w^{\left(\mathcal{I},\phi'\right)}},
\end{array}
\end{equation}
and some algebraic manipulations (reordering the sums and factoring out some terms), the PHD of the SO-GMB turns out to simplify as follows:
\begin{equation}
v(x_1) =\sum_{\mathcal{I}\in\mathcal{F}(\mathbb{I})}
\sum_{\phi\in\Phi}w^{\left(\mathcal{I},\phi\right)}\sum_{\imath\in\mathcal{I}}p^{\left(\phi\right),\imath}(x_1)
\end{equation}
which matches the PHD of original GMB density given by equation~(\ref{GMB_phd}).
\end{proof}

\begin{Rem}
 \normalfont{
PHD is usually perceived as the density of targets, visualizing how their likely places are scattered in the single-target state space. This makes the PHD an instrumental characteristic of multi-target distributions, to the extent that the closest Poisson density to any given density (in terms of Kullback-Leibler distance) is the Poisson density that matches the PHD of the original density. Furthermore, cardinality distribution characterizes the entire probability distribution of the number of targets, using which an EAP or MAP estimate of cardinality can be directly computed~\cite{mahler2007statistical}. Matching the cardinality distributions of two multi-object densities means that they share the same EAP and MAP cardinality estimates. Hence, matching both the PHD and cardinality distributions of two multi-object densities is expected to result in densities that are sufficiently close to each other for estimation and tracking purposes.
 }
\end{Rem}

\begin{Rem}
	\normalfont{
 Note that the SO-GMB density (\ref{SO_GMB}) is also a GMB RFS, especially in the form of the unlabeled version of LMB RFS \cite{wang2015distributed}, and provides the prerequisite condition for distributed fusion with GMB distribution. But the number of terms in SO-GMB distribution is substantially lower than GMB distribution. Indeed,
 the number of components $\left(\hat{w}^{\left(\mathcal{I}\right)},\hat{p}^{\imath}\left(x\right)\right)$ in SO-GMB which need to be stored and computed is $|\mathcal{F}(\mathcal{I})|$ which is substantially smaller than the number of components $\left(w^{(\mathcal{I},\phi)},p^{\left(\phi\right),\imath}(x)\right)$ in the original GMB given by $|\mathcal{F}(\mathcal{I})\times\Phi|$ for $w^{(\mathcal{I},\phi)}$ plus $|\Phi|$ for $p^{\left(\phi\right),\imath}(x)$.

 Before the implementation of the GCI fusion, one need to firstly compute the $\hat{w}^{\left(\mathcal{I}\right)}$ and $\hat{p}^{\imath}\left(x\right)$ under each hypothesis $\mathcal{I}$ according to (\ref{SO_GMB_1}) and (\ref{SO_GMB_2}), which is similar to the marginalization with respect to the association histories performed in M$\delta$-GLMB .
 }
\end{Rem}
\subsection{GCI Fusion Based on SO-GMB Density }
\begin{Def}
A fusion map (for the current time) is a function $\tau:\mathbb{I}_1\rightarrow\mathbb{I}_2$ such that $\tau(i)=\tau(i^*)$ implies $i=i^*$, the set of all such fusion maps is called fusion map space denoted by $\mathcal{T}$. The subset of $\tau$ with domain $\mathcal{I}$ is denoted by $\mathcal{T}(\mathcal{I})$.
\end{Def}
\begin{Rem}
 \normalfont{
The fusion maps play the same role of the measurement-track association map in the Vo-Vo filter~\cite{vo2014labeled}, but note that the fusion maps require one set of tracks in radar node 2 has the same cardinal number in radar node 1.
}
\end{Rem}
\begin{Pro}
The EMD $\pi_\omega(X)$ of the two SO-GMB distributions in (\ref{SO_GMB}), can be approximated as a GMB distribution of the form
\begin{equation}\label{SO_GMB_fusion}
\begin{split}
\pi_\omega&(\left\{x_1,\ldots,x_n\right\})\approx \\ &\sum_\sigma\sum_{(\mathcal{I},\tau)\in\mathcal{F}_n(\mathbb{I}_1)\times\mathcal{T}(\mathcal{I})}w_{\omega}^{(\mathcal{I},\tau)}\prod_{i=1}^{n}p_\omega^{(\tau),\mathcal{I}^v(i)}(x_{\sigma(i)})
\end{split}
\end{equation}
where
\begin{equation}\label{SO_GMB_fusion1}
\begin{split}
w_{\omega}^{(\mathcal{I},\tau)}=\frac{\bar{w}_{\omega}^{(\mathcal{I},\tau)}}
{\displaystyle{\sum_{(\mathcal{I},\tau)\in\mathcal{F}(\mathbb{I}_1)\times\mathcal{T}(\mathcal{I})}\bar{w}_{\omega}^{(\mathcal{I},\tau)}}}
\end{split}
\end{equation}
\begin{equation}\label{SO_GMB_fusion2}
\begin{split}
\bar{w}_{\omega}^{(\mathcal{I},\tau)}=\hat{w}_1(\mathcal{I})^{\omega_1}\hat{w}_2(\tau{(\mathcal{I}))}^{\omega_2}
\prod_{\imath\in\mathcal{I}}\int\hat{p}_1^\imath(x)^{\omega_1}\hat{p}_2^{\tau(\imath)}(x)^{\omega_2}dx
\end{split}
\end{equation}
\begin{equation}\label{SO_GMB_fusion3}
\begin{split}
p_\omega^{(\tau),\imath}(x)=\frac{\hat{p}_1^\imath(x)^{\omega_1}\hat{p}_2^{\tau(\imath)}(x)^{\omega_2}}
{\int\hat{p}_1^\imath(x)^{\omega_1}\hat{p}_2^{\tau(\imath)}(x)^{\omega_2}dx},\,\,\,\,\imath\in\mathcal{I}
\end{split}
\end{equation}
\begin{equation}\label{SO_GMB_fusion4}
\begin{split}
\hat{w}(\mathcal{I})=\hat{w}^{(\mathcal{I})}
\end{split}
\end{equation}
\end{Pro}
%%\begin{proof}
%%Owing to space constraints, the proof process is omitted.
%%\end{proof}
\begin{proof}
By substituting (\ref{SO_GMB}) into the term $\hat{\pi}_s(X)^{\omega_s}$ of (\ref{GCI}), where the subscript $s=1,2$, means the sequence number of sensors, we obtain
\begin{equation}\label{fusion1}
\begin{split}
\hat{\pi}_s&(\left\{x_1,\ldots,x_n\right\})^{\omega_s} \\
&=\left(\sum_\sigma\sum_{\mathcal{I}\in\mathcal{F}_n(\mathbb{I})}\hat{w}_s^{\left(\mathcal{I}\right)}
\prod_{i=1}^{n}\hat{p}_s^{\mathcal{I}^v\left(i\right)}\left(x_{\sigma(i)}\right)\right)^{\omega_s} \\
&=\left(\sum_\sigma\sum_{\mathcal{I}\in\mathcal{F}_n(\mathbb{I})}\hat{w}_s{\left(\mathcal{I}\right)}
\prod_{i=1}^{n}\hat{p}_s^{\mathcal{I}^v\left(i\right)}\left(x_{\sigma(i)}\right)\right)^{\omega_s}
\end{split}
\end{equation}
Motivated by \cite{wang2015distributed,battistelli2013consensus,julier2006empirical}, we use the approximation
\begin{equation}\label{fusion2}
\begin{split}
\left(\sum_id_i\right)^\omega\approx\sum_id_i^\omega
\end{split}
\end{equation}
then (\ref{fusion1}) can be rewritten as
\begin{equation}\label{fusion3}
\begin{split}
\hat{\pi}_s&(\left\{x_1,\ldots,x_n\right\})^{\omega_s} \\
&=\sum_\sigma\sum_{\mathcal{I}\in\mathcal{F}_n(\mathbb{I})}\hat{w}_s{\left(\mathcal{I}\right)}^{\omega_s}
\prod_{i=1}^{n}\hat{p}_s^{\mathcal{I}^v\left(i\right)}\left(x_{\sigma(i)}\right)^{\omega_s}.
\end{split}
\end{equation}
By substituting (\ref{fusion3}) into the numerator of (\ref{GCI}) and utilizing Definition 2, we obtain
%%\begin{equation}\label{fusion4}
%%\begin{split}
%%\hat{\pi}_1&(\left\{x_1,\ldots,x_n\right\})^{\omega_1}\hat{\pi}_2(\left\{x_1,\ldots,x_n\right\})^{\omega_2}  \\
%%&=\sum_{\sigma}\sum_{\mathcal{I}\in\mathcal{F}_n(\mathbb{I}_1)}\sum_{\tau\in\mathcal{T}(\mathcal{I})}\hat{w}_1(\mathcal{I})^{\omega_1}\prod_{i=1}^{n}\hat{p}_1^{\mathcal{I}^v(i)}(x_{\sigma(i)})^{\omega_1}\\
%%&\,\,\,\,\,\,\,\,\hat{w}_2(\tau{(\mathcal{I})})^{\omega_2}\prod_{i=1}^{n}\hat{p}_2^{\tau{(\mathcal{I}^v(i))}}(x_{\sigma(i)})^{\omega} \\
%%&=\sum_{\sigma}\sum_{\mathcal{I}\in\mathcal{F}_n(\mathbb{I}_1)}\sum_{\tau\in\mathcal{T}(\mathcal{I})}\hat{w}_1(\mathcal{I})^{\omega_1}\hat{w}_2(\tau{(\mathcal{I})})^{\omega_2}\\
%%&\,\,\,\,\,\,\,\,\prod_{\imath\in\mathcal{I}}\int\hat{p}_1^{\imath}(x_{\sigma(i)})^{\omega_1}\hat{p}_2^{\tau{(\imath)}}(x_{\sigma(i)})^{\omega_2}dx_i \\
%%&\,\,\,\,\,\,\,\,\prod_{i=1}^{n}\frac{\hat{p}_1^{\mathcal{I}^v(i)}(x_{\sigma(i)})^{\omega_1}\hat{p}_2^{\tau{(\mathcal{I}^v(i))}}(x_{\sigma(i)})^{\omega_2}}{\int\hat{p}_1^{\mathcal{I}^v(i)}
%%(x_{\sigma(i)})^{\omega_1}\hat{p}_2^{\tau{(\mathcal{I}^v(i))}}(x_{\sigma(i)})^{\omega_2}dx_i}\\
%%&=\sum_{\sigma}\sum_{\mathcal{I}\in\mathcal{F}_n(\mathbb{I}_1)}\sum_{\tau\in\mathcal{T}(\mathcal{I})}\bar{w}_{\omega}^{(\mathcal{I},\tau)}\prod_{i=1}^{n}p_\omega^{(\tau),\mathcal{I}^v(i)}(x_{\sigma(i)})
%%\end{split}
%%\end{equation}
\begin{equation}\label{fusion44}
\begin{split}
\hat{\pi}_1&(\left\{x_1,\ldots,x_n\right\})^{\omega_1}\hat{\pi}_2(\left\{x_1,\ldots,x_n\right\})^{\omega_2}  \\
&=\sum_{\sigma}\sum_{\mathcal{I}\in\mathcal{F}_n(\mathbb{I}_1)}\sum_{\tau\in\mathcal{T}(\mathcal{I})}\hat{w}_1(\mathcal{I})^{\omega_1}\prod_{i=1}^{n}\hat{p}_1^{\mathcal{I}^v(i)}(x_{\sigma(i)})^{\omega_1}\\
&\,\,\,\,\,\,\,\,\hat{w}_2(\tau{(\mathcal{I})})^{\omega_2}\prod_{i=1}^{n}\hat{p}_2^{\tau{(\mathcal{I}^v(i))}}(x_{\sigma(i)})^{\omega} \\
\end{split}
\end{equation}
\begin{equation*}
\begin{split}
&=\sum_{\sigma}\sum_{\mathcal{I}\in\mathcal{F}_n(\mathbb{I}_1)}\sum_{\tau\in\mathcal{T}(\mathcal{I})}\hat{w}_1(\mathcal{I})^{\omega_1}\hat{w}_2(\tau{(\mathcal{I})})^{\omega_2}\\
&\,\,\,\,\,\,\,\,\prod_{\imath\in\mathcal{I}}\int\hat{p}_1^{\imath}(x_{\sigma(i)})^{\omega_1}\hat{p}_2^{\tau{(\imath)}}(x_{\sigma(i)})^{\omega_2}dx_i \\
&\,\,\,\,\,\,\,\,\prod_{i=1}^{n}\frac{\hat{p}_1^{\mathcal{I}^v(i)}(x_{\sigma(i)})^{\omega_1}\hat{p}_2^{\tau{(\mathcal{I}^v(i))}}(x_{\sigma(i)})^{\omega_2}}{\int\hat{p}_1^{\mathcal{I}^v(i)}
(x_{\sigma(i)})^{\omega_1}\hat{p}_2^{\tau{(\mathcal{I}^v(i))}}(x_{\sigma(i)})^{\omega_2}dx_i}\\
&=\sum_{\sigma}\sum_{\mathcal{I}\in\mathcal{F}_n(\mathbb{I}_1)}\sum_{\tau\in\mathcal{T}(\mathcal{I})}\bar{w}_{\omega}^{(\mathcal{I},\tau)}\prod_{i=1}^{n}p_\omega^{(\tau),\mathcal{I}^v(i)}(x_{\sigma(i)})
\end{split}
\end{equation*}
Thus the denominator of (\ref{GCI}) can be computed as:
\begin{equation}\label{fusion5}
\begin{split}
\int&\hat{\pi}_1(\left\{x_1,\ldots,x_n\right\})^{\omega_1}\hat{\pi}_2(\left\{x_1,\ldots,x_n\right\})^{\omega_2}\delta X \\
&=\sum_{n=0}^{\infty}\frac{1}{n!}\sum_{\sigma}\sum_{\mathcal{I}\in\mathcal{F}_n(\mathbb{I}_1)}\sum_{\tau\in\mathcal{T}(\mathcal{I})}\bar{w}_{\omega}^{(\mathcal{I},\tau)}\\
&=\sum_{\mathcal{I}\in\mathcal{F}(\mathbb{I}_1)}\sum_{\tau\in\mathcal{T}(\mathcal{I})}\bar{w}_{\omega}^{(\mathcal{I},\tau)}.
\end{split}
\end{equation}
Finally, by substituting (\ref{fusion44}) and (\ref{fusion5}) into (\ref{GCI}), we obtain the fused density as the form of (\ref{SO_GMB_fusion}).
\end{proof}
\begin{Rem}
 \normalfont{
The EMD of the two SO-GMB distributions turns out to be another GMB distribution, which can allow the subsequent fusion with another radar node. However, it can be seen from (\ref{SO_GMB_fusion}) that after fusion, each hypothesis $\mathcal{I}\in\mathcal{F}_n(\mathbb{I})$ generates a set of $|\mathcal{T}(\mathcal{I})|$ fusion maps and then we perform all permutations on the number of targets $n$, resulting in an approximate computational complexity of $O\left(|\mathcal{T}(\mathcal{I})|\times n!\right)$. In order to reduce the cost of computation, one can perform truncation of fused GMB density using the ranked assignment strategy~\cite{vo2014labeled, murty1968letter}.
}
\end{Rem}

%\subsection{Summary of the GCI with SO-GMB Distributions}
The schematic diagram shown in Fig.~\ref{framework} shows a graphical representation of the overall distributed fusion process as proposed in this paper. Assuming that at each node of the sensor network, a Vo-Vo filter is locally run, at each time $k$, a local GLMB posterior is computed in each node using the measurement (e.g. the radar scan) acquired by the sensor in that node. Each node also receives GLMB posteriors from the other nodes which are connected to it in the network. The proposed method is then applied locally in each node to fuse the local and all received posteriors.
\begin{figure}
\begin{center}
\includegraphics[width=0.85\columnwidth,draft=false]{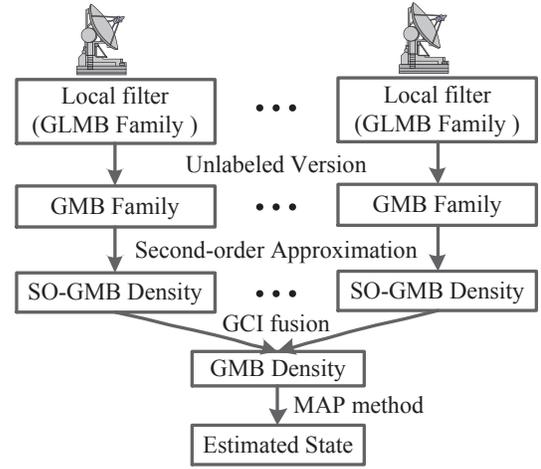}
\end{center}
\caption{Distributed fusion with SO-GMB filter schematic.}
\label{framework}
\end{figure}

The first step of the distributed fusion procedure involves turning all the GLMB posteriors to their equivalent GMBs by marginalizing the labels. Each GMB density is then replaced with its second-order approximation (SO-GMB) which preserves both the PHD and cardinality distribution. The local and incoming posteriors are then sequentially fused using to the GCI-rule formulated for two GMB densities (see Proposition~2). The resulting GMB is expected to sub-optimally encapsulate all the complementary information provided by the neibouring sensors. Thus, a highly accurate multi-object state estimate (MAP  or EAP) can be obtained from the cardinality probability mass function and the location PDFs. As implied by Fig.~\ref{framework}, in our experiments we computed MAP estimates.

\begin{Rem}
\normalfont
Each fused SO-GMB density can also be used to produce trajectories. Once we choose a local radar node and use its label space as pre-image in the fusion map, the fused SO-GMB distribution can inherit and reserve its labels.
\end{Rem}

\section{Simulation Results and Discussion}
In two challenging scenarios, we examined the performance of the proposed GCI-based fusion using SO-GMB distributions, and compared the tracking outcomes (in terms of the resulting OSPA errors~\cite{schuhmacher2008consistent}) with the recent GCI-based fusion rule using FO-GMB approximations. With both methods, the local filters are Vo-Vo filters and the fusion weight of each node $\omega_1,\omega_2$ in (\ref{GCI}) are both chosen as 0.5. All performance metrics are averaged over 100 Monte Carlo runs.

In both scenarios, the single target state includes planar position and velocity, $$x_k=[p_{x,k}\ \ p_{y,k}\ \ \dot{p}_{x,k}\ \ \dot{p}_{y,k}]^\top$$
and each non-clutter point measurement is a noisy version of the planar position of a single-target,
$$z_k=[z_{x,k}\ \ z_{y,k}]^\top.$$
Thus, the single-target measurement model is given by
\begin{eqnarray}
g_k(z_k|x_k) &=& \mathcal{N}(z_k;H_k x_k,R_k) \\
H_k &=& \left[
  \begin{array}{ccc}
    I_2 & 0_2
  \end{array}
\right]
\\
R_k &=& \sigma_{\varepsilon}^2I_2
\end{eqnarray}
where $I_n$ and $0_n$ denote the $n \times n$ identity and zero matrices, respectively, and $\Delta=1$~s is the sampling period.

To model single-target motions, the nearly constant velocity model with the following state transition density is used,
\begin{eqnarray}
f_{k|k-1}(x_k|x_{k-1}) &=& \mathcal{N}(x_k;F_kx_{k-1},Q_k) \\
F_k &=& \left[
  \begin{array}{ccc}
    I_2 & \Delta I_2 \\
    0_2 & I_2 \\
  \end{array}
\right] \\
Q_k &=& \sigma_v^2\left[
  \begin{array}{ccc}
    \frac{\Delta^4}{4}I_2 & \frac{\Delta^3}{2}I_2 \\
    \frac{\Delta^3}{2}I_2 & \Delta^2I_2 \\
  \end{array}
\right]
\end{eqnarray}

The process and measurement noise powers, $\sigma_v^2$ and $\sigma_\varepsilon^2$, are fixed parameters. In our simulations the chosen values are $\sigma_v = 5~\mathrm{m/s}^2$ and $\sigma_{\varepsilon} = 14$~m. The survival probability is $P_{S,k} = 0.98$; target detection in each sensor is independent of the others, and probability of detection at all sensors is $P_D = 0.9$. The number of clutter reports in each scan is Poisson distributed with $\lambda_c = 15.$ Each clutter report is sampled uniformly over the whole surveillance region.

\subsection{Scenario 1}
A set of targets move in the two dimensional region $[0,10000~\mathrm{m}]\times[0,10000~\mathrm{m}]$, all traveling in straight paths with different but constant velocities. The number of targets is time varying due to births and deaths. Trajectories of two targets intersect at  location $(4000~\mathrm{m},2200~\mathrm{m})$ and time $k=5$~s, three targets meet at location $(6000~\mathrm{m},4900~\mathrm{m})$ and time $k=20$~s. The region and tracks are shown in Fig.~\ref{fig: location1}.

\begin{figure}
\begin{center}
\includegraphics[width=0.85\columnwidth,draft=false]{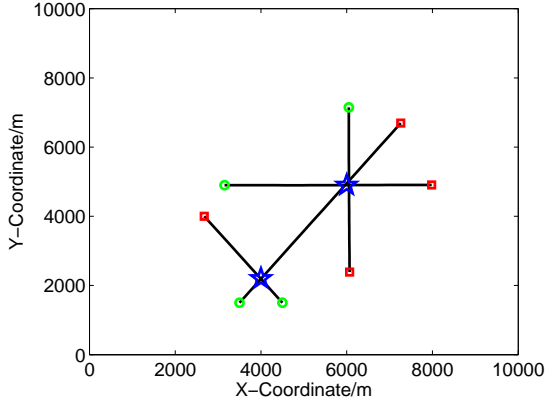}
\end{center}
\caption{Target trajectories considered in the simulation experiment. The start/end point for each trajectory is denoted, respectively, by $\circ|\square$.
The $\star$ indicates a rendezvous point.}
\label{fig: location1}
\end{figure}

The birth model is a LMB RFS with four components, all sharing the same probability of existence of $r_B^{(i)}=0.06$, but having four different Gaussian densities. The Gaussians have the same covariance matrix of $P_B=\mathrm{diag}(100^2,100^2,20^2,20^2)$ but different means,
\begin{equation}
\nonumber
\begin{array}{rcl}
\vspace{2mm}
m_B^{(1)} &=& [3500\ \ 1500\ \ 0\ \ 0]^\top\\
\vspace{2mm}
m_B^{(2)} &=& [4500\ \ 1500\ \ 0\ \ 0]^\top\\
\vspace{2mm}
m_B^{(3)} &=& [3150\ \ 4900\ \ 0\ \ 0]^\top\\
\vspace{2mm}
m_B^{(4)} &=& [6050\ \ 7150\ \ 0\ \ 0]^\top.
\end{array}
\end{equation}

Comparative results in terms of cardinality statistics and OSPA error (with parameters $c = 200$~m, $p = 2$) are shown in Figs.~\ref{fig:card1} and~\ref{fig:ospa1}. In each plot, the results returned by each of the two local Vo-Vo filters are presented along with those resulted from GCI-fusion of FO-GMB densities and from our proposed GCI-based fusion of SO-GMB densities.

From Fig.~\ref{fig:card1}, we observe that in this scenario, all the filters (both local and fusion-based ones) perform similarly in terms of cardinality estimation errors. Figure~\ref{fig:ospa1} shows that the advantage of sensor fusion is more evident in OSPA errors which are substantially smaller with GCI-based fusion methods. However, in this scenario, no substantial improvement seems to be obtained with GCI-based fusion of SO-GMB densities (proposed in this paper) compared to fusion of FO-GMBs. Indeed, a closer look at Fig.~\ref{fig:ospa1} would reveal that OSPA errors returned by fusion of SO-GMBs are visibly lower than the ones returned by FO-GMBs at the times of death or when the targets intersect (are located in close proximity of each other).

\begin{figure}
\begin{center}
\includegraphics[width=0.85\columnwidth,draft=false]{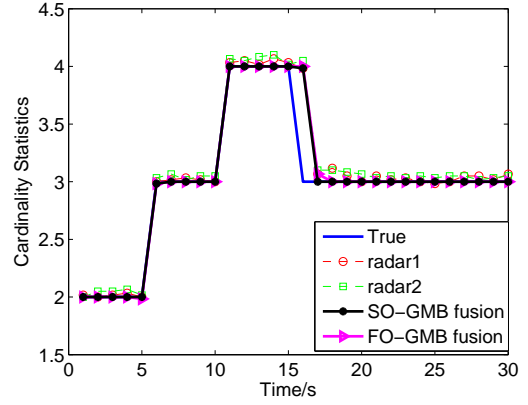}
\end{center}
\caption{Cardinality statistics returned by the local Vo-Vo filter at node 1, the local Vo-Vo filter at node 2, FO-GMB fusion, and SO-GMB fusion in scenario 1. The plotted results are the average of 100 Monte Carlo runs.}
\label{fig:card1}
\end{figure}

\begin{figure}
\begin{center}
\includegraphics[width=0.85\columnwidth,draft=false]{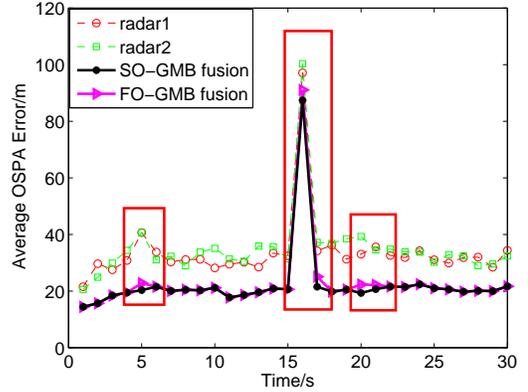}
\end{center}
\caption{OSPA errors returned by the local Vo-Vo filter at node 1, the local Vo-Vo filter at node 2, FO-GMB fusion, and SO-GMB fusion in scenario 1. The plotted results are the average of 100 Monte Carlo runs.}
\label{fig:ospa1}
\end{figure}

\subsection{Scenario 2}
In order to signify the advantages gained by SO-GMB fusion, triggered by the observation exclaimed in the previous paragraph, scenario 2 was designed to include tracks that are in close proximity most of the times, and include deaths and births, as shown in Fig.~\ref{fig:location4}.

\begin{figure}
\begin{center}
\includegraphics[width=0.85\columnwidth,draft=false]{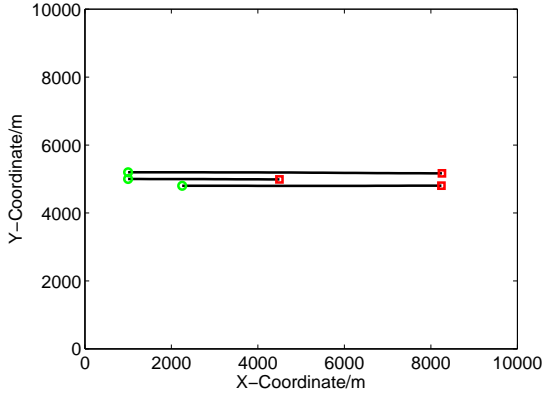}
\end{center}
\caption{Target trajectories in scenario 2. The start/end point for each trajectory is denoted, respectively, by $\circ|\square$.}
\label{fig:location4}
\end{figure}

The true and estimated cardinalities by the FO-GMB fusion method, along with the standard deviation of the estimates over 100 Monte Carlo runs, are together presented in Fig.~\ref{fig:focard2}. Similar plots for SO-GMB fusion cardinality estimates are shown in Fig.~\ref{fig:socard2}. Comparing the standard deviations of cardinality estimates in the two figures would lead to observation that our proposed distributed sensor fusion method outperforms the recent FO-GMB fusion method in terms of having a smaller standard deviation.

\begin{figure}
\begin{center}
\includegraphics[width=0.85\columnwidth,draft=false]{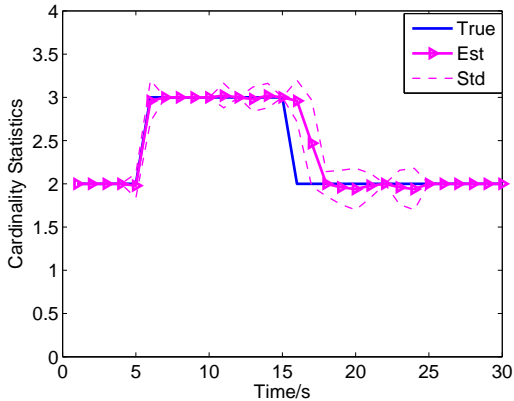}
\end{center}
\caption{Cardinality statistics returned by the FO-GMB fusion method in scenario 2. The plotted results are the average of 100 Monte Carlo runs.}
\label{fig:focard2}
\end{figure}

\begin{figure}
\begin{center}
\includegraphics[width=0.85\columnwidth,draft=false]{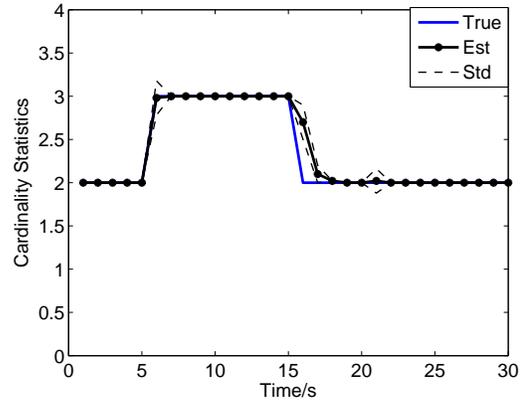}
\end{center}
\caption{Cardinality statistics returned by the SO-GMB fusion method in scenario 2. The plotted results are the average of 100 Monte Carlo runs.}
\label{fig:socard2}
\end{figure}
Figure~\ref{fig:ospa2} shows the average OSPA errors returned by the two fusion methods in this challenging scenario. Again, substantial reduction in estimation error is observed with SO-GMB fusion compared to FO-GMB fusion. Overall, the fusion method proposed in this paper appears to perform with better stability and more accurately and is more suited to tackle the problems introduced by the cardinality changes.

\begin{figure}
\begin{center}
\includegraphics[width=0.85\columnwidth,draft=false]{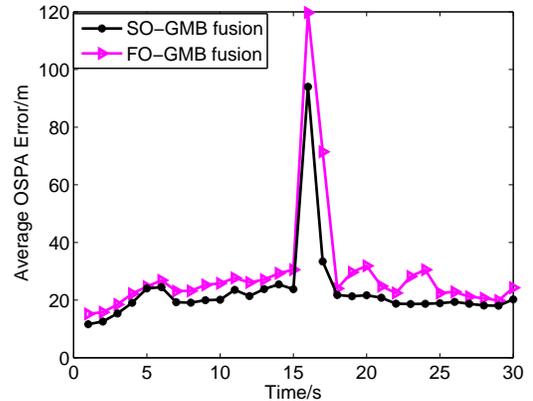}
\end{center}
\caption{Comparison of OSPA errors returned by FO-GMB fusion and SO-GMB fusion in scenario 2. The plotted results are the average of 100 Monte Carlo runs.}
\label{fig:ospa2}
\end{figure}

\section{Conclusion}
In this paper, we address the problem of distributed multi-target tracking with labeled set filters in the framework of GCI with consideration of label space mismatching phenomenon. Based on the notation of GMB family, firstly, we propose an efficient approximation to the GMB family which preserves both the PHD and cardinality distribution, referred to as second-order approximation of GMB (SO-GMB) density. Then, we derive the explicit formula for GCI-based fusion of SO-GMB densities and devise a sequential fusion method for distributed target tracking in sensor networks such as bistatic radar networks. Finally, we compare the recently developed method of GCI-based fusion of  first-order approximation GMB (FO-GMB) densities with our proposed method in two scenarios. The results show that while both methods perform similar in simple tracking situations with no frequent targets deaths and intersections between target trajectories, the SO-GMB fusion method is advantageous in more challenging applications where targets move in close proximity with frequent deaths.

\section*{Acknowledgment}
This work was supported by the Australian Research Council's Discovery Project Program, via the ARC Discovery Project grants DP130104404 and DP160104662, and supported by the National Natural Science Foundation of China under Grants 61301266, the Chinese Postdoctoral Science Foundation under Grant 2014M550465.

\bibliographystyle{ieeetr}
\bibliography{SOGMB}
%\end{CJK}
\end{document}